\documentclass[11pt]{article}

\usepackage{amssymb,amsmath,amsfonts,latexsym,bbm,ae,aecompl}
\usepackage{color,graphicx,enumitem,subfigure}

\usepackage{array}
\usepackage{makecell}

\newcommand{\cA}{{\mathcal A}}

\newcommand{\cD}{{\mathcal D}}

\newcommand{\cH}{{\mathcal H}}

\newcommand{\cS}{{\mathcal S}}
\newcommand{\cT}{{\mathcal T}}

\newcommand{\pro}{{\mathit P}}
\newcommand{\re}{{\mathit R}}

\newcommand{\WC}{\mathit{WC}}

\newcommand{\hT}{{\mathit TO}}

\newcommand{\qed}{\hfill $\square$ \smallbreak}
\newenvironment{proof}{\noindent{\bf Proof:}}{\qed}

\newtheorem{theorem}{Theorem}
\newtheorem{lemma}{Lemma}
\newtheorem{definition}{Definition}

\newcommand{\ignore}[1]{}

\newlength{\pagewidth}
\begin{document}

\baselineskip 	3ex
\parskip 		1ex

\title{Universal Routing in Multi-hop Radio Networks\thanks{A preliminary version of this paper was published in~\cite{DBLP:conf/dialm/ChlebusCK14}} \vfill}

\author{	Bogdan S. Chlebus\,\footnotemark[1] 
		\and
		Vicent Cholvi\,\footnotemark[2]
		\and
		Dariusz R. Kowalski\footnotemark[3]}

\footnotetext[1]{
		Department of Computer Science and Engineering,
             	University of Colorado Denver,
		Denver, Colorado, USA. 
		Supported by the NSF Grant 1016847.}

\footnotetext[2]{
		Department of Computer Science, 
		Universitat Jaume I, 
		Castell\'on, Spain.
		Supported by the CICYT Grant TIN2011-28347-C02-01.}
			
\footnotetext[3]{
		Department of Computer Science,
            	University of Liverpool,
            	Liverpool, UK. 
		Supported by the EPSRC grant EP/G023018/1. \vfill}

\date{\today}

\maketitle

\vfill

\begin{abstract}
In this article we introduce a new model to study stability in multi-hop wireless networks in the framework of adversarial queueing. In such a model, a routing protocol consists of three components: a transmission policy, a scheduling policy to select the packet to transmit form a set of packets parked at a node, and a hearing control mechanism to coordinate transmissions with scheduling.  For such a setting, we propose a definition of universal stability that takes into account not only the scheduling policies (as in the standard wireline adversarial model), but also the transmission policies.

First, we show that any scheduling policy that is unstable in the classical wireline adversarial model remains unstable in the multi-hop radio network model, even in scenarios free of interferences. Then, we show that both SIS and LIS (two well-known universally stable scheduling policies in the wireline adversarial model) remain stable in the multi-hop radio network model, provided a proactive hearing control is used. In contrast, such scheduling policies turn out to be unstable when using a reactive hearing control. However, the scheduling policy LIS can be enforced to be universally stable provided ties are resolved in a permanent manner. Such a situation doesn?t hold in the case of SIS, which remains unstable regardless of how ties are resolved. Furthermore, for some transmission policies which we call regular, we also show that all scheduling policies that are universally stable when using proactive hearing control (which include SIS and LIS) remain universally stable when using reactive hearing control.

\vfill


\end{abstract}

\section{Introduction}

Wireless data communication involves multiple technologies  interacting with each other.
In order to study communication algorithms for wireless networks, one needs models that abstract from incidental details and capture the essential aspects of wireless networking.

A node of a wireless network can transmit messages within its transmission range. 
Such a range is determined by the power of the transmitting device and the surrounding topography. 
A possible approach to model this is through geometric wireless networks in which ranges are determined by distances assigned to nodes.
A popular special case of geometric networks has all the ranges equal, so that the topologies of such networks are unit-disk graphs; see~\cite{EmekGKPPS09,KuhnWZ-WN08,Urrutia02}.
Another alternative is a signal-to-interference-plus-noise ratio (SINR) model which incorporates interference and noise explicitly in determining the ranges of reliable transmissions.
In a SINR setting, a transmission is successful when a suitable ratio of ``good'' to ``bad'' components of a received signal is above a  threshold; see~\cite{DaumGKN13,JurdzinskiKRS-DISC13,JurdzinskiKS13,RickenbachWZ09}.

Radio data networks model wireless communication in which just one channel is used for transmissions.
Signal receptions at a node that overlap in time interfere with one another, so that none can be received successfully.
The feasibility of a node-to-node direct transmissions determines how nodes can reach each other; 
reachability, so defined, is a relation on the nodes of a network.
This suggests modeling the topology of a multi-hop network as an arbitrary connected graph where edges represent direct reachability; see~\cite{ChlebusKP-OPODIS12,ChlebusKPR-ICALP11}.

The model used in this paper abstract from geometrical constraints imposed on ranges of transmissions and represents the relation of reachability as a general graph.
We consider simple graphs, that is, with symmetric bi-directional edges; see~\cite{KowalskiP-DC05}.
This is to make a dialog feasible between pairs of nodes that can communicate by direct transmissions.

Dynamic store-and-forward routing in wireless networks differs from the respective routing in wireline networks.
On wireline networks, a \emph{scheduling policy}, that manages queues of packets at nodes, is the only essential component of source routing when packets have their routes determined from the point of injection.
This is because a node can transmit a packet per round over any outgoing link, and simultaneously accept a packet per round over any incoming link, these packets coming and going simultaneously.
In wireless networks, coordinating timings of transmissions among nodes, with the goal to avoid collisions resulting from receiving overlapping transmissions, has the potential to improve performance of routing.
Such coordination is handled by a \emph{transmission policy}, which is another essential component of a routing protocol.
These two policies need to cooperate with each other, which is delegated to a \emph{hearing control mechanism}.

The structuring of routing protocols we consider can be related to the OSI model.
Protocol stacks separate subtasks of a communication algorithm into layers of their individual functionalities.
Coordinating message transmissions between nodes has the purpose to avoid collisions of messages, so it can be considered to belong to the medium-access control sublayer of the data-link layer.
A mechanism of exchanging control messages to provide hearing control belongs to the logical-link control sublayer of the data-link layer.
A scheduling policy can be considered as operating at the network layer.

Cross-layer approach proposes to relax these functionality specifications to enhance efficiency.
Such a relaxation is accomplished by providing additional interactions between layers.
See~\cite{FoukalasGA08,Jurdak07,LinSS06,ShakkottaiRK03} for more about the motivation and guidelines in developing cross-layered algorithms.

\paragraph{Related work.}
Adversarial queuing in wireline networks, as a methodology to study stability in the worst case abstracting from stochastic assumptions on traffic generation, was initiated by Borodin et al.~\cite{BorodinKRSW01} and Andrews et al.~\cite{AndrewsAFLLK01}. Such methodology, which is also referred as \emph{Adversarial Queuing Theory} ($AQT$), considers the time evolution of a packet-routing network as a game between a malicious \emph{adversary} that has the power to perform a number of actions (such as injecting packets at particular nodes, choosing their destination, routing them, etc.) and the underlying system.  That adversary, based on its knowledge of the behavior of the system, can devise the scenario that maximizes the ``stress'' on the system. They also introduced the notions of universal stability of protocols and networks in that setting, defined as the property of keeping the amount of traffic  in a system always bounded over time.

Since then, much research has been carried out to gain an understanding of the factors that affect the stability of packet-switched networks (see~\cite{CholviE07}). A systematic account of issues related to universal stability in adversarial routing  was given by {\`A}lvarez et al.~\cite{AlvarezBS-SICOMP04}. 
In~\cite{DBLP:journals/join/BorodinOR04}, Borodin et al. consider a scenario in which links can have \emph{slowdowns} in the transmission of packets, and variations in link capacities (but not both). Networks with nodes and links that occasionally fail were studied by {\`A}lvarez et al.~\cite{DBLP:journals/tcs/AlvarezBS11}. Networks with bandwidth and delay parameters associated with links were considered in Blesa et al.~\cite{BlesaCFLMSST09} and Borodin et al.~\cite{BorodinOR04}. Such behavior can be considered as capturing some properties of wireless networks.

Routing in radio networks was considered as early as in Gitman et al. \cite{GitmanSF76}. In the context of the adversarial queueing, Chlebus et al.~\cite{Chlebus:2012:AQM:2071379.2071384} have some stability results in the multiple access channel. Stability was studied in general wireless networks without interferences by Andrews and Zhang \cite{AndrewsZ-JACM05,AndrewsZ07}, and by Cholvi and Kowalski~\cite{CholviK10}. Andrews et al.~\cite{DBLP:journals/ton/LimJA14} analyzed the stability of the max-weight protocol in wireless networks in scenarios with interferences, but assuming the existence of a set of feasible edge rate vectors sufficient to keep the network stable.

\paragraph{Our results.}
In this article we study dynamic routing in multi-hop radio networks in the framework of adversarial queueing. 

We consider cross-layer interactions of the following three components of routing protocols: a \emph{transmission policy} for medium-access control, a \emph{scheduling policy} to select the packet to transmit form a set of packets parked at a node, and \emph{hearing control mechanisms} to coordinate transmissions with scheduling. Furthermore, the injection of packets is delegated to an adversary  that specifies for each packet its complete itinerary. A difference between our radio adversarial model and the one addressed by Andrews et al.~\cite{DBLP:journals/ton/LimJA14}  is that we incorporate collisions into the model.

For such a setting, we propose a definition of universal stability that takes into account not only the scheduling policies (as in the standard wireline adversarial model), but also how transmission policies are coordinated with scheduling.

Our first result is concerned with instability. We show that, any scheduling policy that is unstable in the classical wireline adversarial model~\cite{AndrewsAFLLK01,BorodinKRSW01} remains unstable in the multi-hop radio network model, even in scenarios free of interferences.

From the point of view of stability, we mainly focus on SIS and LIS, which are two well-known universally stable scheduling policies in the wireline adversarial model~\cite{AndrewsAFLLK01}. Such policies exemplify how the stability of a routing protocol can be affected by both the \emph{hearing control} and the \emph{transmission policies} (concepts that are formally introduced in the next section). We show that both SIS and LIS remain stable  in the multi-hop radio network model, provided  a \emph{proactive} hearing control is used. In contrast, such scheduling policies turn out to be unstable when using a \emph{reactive} hearing control. However, the scheduling policy LIS can be enforced to be universally stable provided ties are resolved in a \emph{permanent} manner. Such a situation doesn't hold in the case of SIS, which remains unstable regardless of how ties are resolved. Furthermore, for some transmission policies which we call \emph{regular}, we also show that all scheduling policies that are universally stable when using  proactive hearing control (which include SIS and LIS) remain universally stable when using reactive hearing control.

The rest of the paper is structured as follows. In Section~\ref{sec:preliminaries} we introduce the formal model for multi-hop radio networks. In Section~\ref{sec:general} it is presented a result that is valid in all the considered scenarios, even if they are free of interferences. Sections~\ref{sec:proactive},~\ref{sec:reactive} and~\ref{sec:regular} present some results that show how the stability of a routing protocol is affected by its components (i.e., the scheduling policy, the transmission policy and the hearing control). Finally, the paper ends in Section~\ref{sec:conclusions} with some conclusions.

\section{Technical Preliminaries}
\label{sec:preliminaries}

This section is divided in two parts. In Section~\ref{subsec:model} we introduce our \emph{multi-hop radio network model}. In Section~\ref{subsec:stability} we define stability in the context of the presented model.

\subsection{Model}
\label{subsec:model}

A network is modeled as a simple graph $G=(V,E)$, where $V$ is the set of nodes and $E$ is a set of edges.
A node of the graph represents a  transceiver that can act both as a sender and as a receiver.
An edge~$(u,w)$ represents the property that the nodes~$u$ and~$w$ can transmit directly to each other, in the sense that there are two independent directed links from~$u$ to~$w$ and from~$w$ to~$u$ available for direct transmissions.

There are $n$ nodes in the network.
Each node is assigned a unique name, which is an integer in $[0,n-1]$.
Every node knows $n$ and its own name, in the sense that they can be used as a part of code of protocols.

Nodes have access to local clocks ticking at the same rate.
Time is divided into time intervals of fixed length that we call \emph{rounds}.
Local computations at a node are considered to be of negligible duration.

\subsubsection{Messages}

The contents of transmissions are structured into chunks of data that we call \emph{messages}.
Messages are of two kinds: packets and control messages. 
A \emph{packet} carries a header that includes a destination address followed by this packet's content.
A \emph{control message} carries a string of bits used to coordinate actions among the nodes.
Control messages are significantly shorter than packet ones.
We assume that transmitting a control message takes an insignificant amount of time compared to what is needed to transmit a packet of data.
Rounds are scaled to the amount of time it takes to transmit a message with a packet.
Packets and control messages are interleaved in executions of routing protocols.

\subsubsection{Data transmissions}

A node may transmit exactly one message in a round or pause in this round.
A message received successfully by a node is said to be \emph{heard} by that node.

Radio networks are defined by the following two properties.
First, when two messages arrive at a node $v$ transmitted by its neighbor such that their receipt overlaps in time, then they interfere with each other and none can be heard by~$v$.
Second, when only one neighbor of a node~$v$ transmits a message then $v$ hears this message.

Radio networks are \emph{single port} in the sense that a node can transmit at most one message in a round and hear at most one message in a round.

\subsubsection{Routing}

Packets get injected into nodes to be delivered to their respective destination nodes by traversing paths.
We consider \emph{source routing} in which the entire path of a packet is known at the source where the packet is injected. 
When a packet traverses a link from~$v$ to~$w$ on such a path, by way of $v$ transmitting the packet and $w$ hearing it, then $w$ is the \emph{intended recipient} of the packet transmitted by $v$.
A packet message carries  a header which includes the packet's itinerary.
This allows for a node that hears a packet message to decide how to process the message.
The following are three cases of what occurs when a node~$v$ hears a message.
One case is when $v$ is not the intended recipient of the message: then $v$ discards the message.
The other occurs when $v$ is the final destination for the message: then $v$ consumes the message. 
The final case is when $v$ is the message's  recipient, but not its final destination: then $v$ enqueues the message to wait to be forwarded to a neighbor.

\subsubsection{Routing protocols}

A \emph{routing protocol} manages how packets traverse their respective assigned paths.
We consider distributed routing protocols, in which each node runs its code independently. 
Packets to be forwarded to a node's neighbors may need to wait for their turn to be transmitted.
Each node contains a buffer space to temporarily store data to be transmitted in the future, which is organized as a queue. 
Each node maintains a single queue, rather than a dedicated queue for each outgoing link, as in the wireline adversarial queuing model. The protocols we consider operate under the principle that no packet is discarded until its delivery to the destination.
To make this meaningful, we assume that the buffer space at a node can store an arbitrarily large number of packets, although we want to keep bounded queues at nodes.

Routing protocols make use of three components: the \emph{scheduler}, the \emph{transmission policy} and the \emph{hearing control}: 

\begin{enumerate}
\item
\emph{Scheduler}: a scheduling policy is understood as a rule to select a packet to transmit from a group of packets at a node. When the packet is eventually heard by the corresponding neighbor, then the sending node removes this packet from its queue.

Popular queueing policies include FIFO (First-In-First-Out), Furthest-From-Source (FFS), Furthest-To-Go (FTG), Nearest-To-Source (NTS), Nearest-To-Go  (NTG), Shortest-In-System  (SIS), and Long\-est-In-System  (LIS).

\item
\emph{Transmission policy}: a transmission policy  indicates to a node whether to transmit in a round or rather to pause.
Such an indication is either in the affirmative, meaning ``do transmit in this round,'' or in the negative, meaning ``do not transmit in this round.'' 

It must be taken into account that some transmissions, during an execution of a routing protocol in a radio network, may not be successful in transferring packets to their intended recipients, due to collisions. To mitigate this adverse effect, some nodes may be indicated by the transmission policy to pause in a round to create more room in the only transmitting frequency available for the neighbors to possibly transmit and hear successfully. Therefore, the ultimate goal of a transmission policy is to facilitate packet movement by avoiding collisions of packets transmitted by different neighbors of nodes. 

We abstract from implementing transmission policies, but delegate such a task to a \emph{transmission oracle} (or simply \emph{oracle}) that will indicate each node whether or not to transmit in a round. This approach allows to abstract from specific transmission policies to consider the qualities of scheduling policies independently of transmission policies.

The relevant, desirable functionality such transmission oracles need to provide is that \emph{each link can successfully transmit within a bounded number of rounds, provided there is a packet ready to be transmitted}.

We say that an integer $h_e>0$ is a \emph{hearing latency of link~$e$} if link~$e$ is guaranteed to be able to successfully transmit at least one packet each $h_e$ consecutive rounds; if such a number does not exist then $e$ is said to have an \emph{unbounded link hearing latency}, denoted $h_e=\infty$. 
We say that a transmission oracle \emph{provides a link latency~$h$} when hearing latencies of all links are upper bounded by the number~$h$.
We denote by $\hT_h$ the class of transmission oracles that provide a link hearing latency of~$h$.

In the Appendix~\ref{sec:transmissions} we show how some of the above mentioned oracles can be implemented by using specific transmission policies whose code is imbedded in the code of a routing protocol, which we refer as \emph{implemented transmissions}.

\item
\emph{Hearing control}: what we call hearing control is a mechanism that, in coordination with the scheduling and the transmission policies, is responsible for performing the transmissions of messages. We consider two different possibilities:

\begin{itemize}
\item
\emph{Proactive hearing control} (denoted $\pro$): when a node wants to transmit in a round, the hearing control first obtains a list of neighbors that will hear the message in the round.  In such an arrangement, the scheduler selects a packet from those parked in the queue that have one of these neighbors on their paths to traverse, which is finally transmitted. This hearing control allows to avoid unexpected collisions, but its implementation requires a handshaking which adds a significant overhead (see for instance~\cite{Karn90}). 

\item
\emph{Reactive hearing control} (denoted $\re$): when a node wants to transmit in a round, it does. And immediately after the transmission, a mechanism is invoked (by the hearing control) to detect if the intended recipient node has heard the message. This means that a scheduler learns about the effectiveness of selection after a transmission. A given packet must be retransmitted until it is eventually heard by the corresponding neighbor.  Whereas now transmissions may suffer from collisions, it has the advantage of requiring an overhead much smaller than in the previous case.

\end{itemize}

Figure~\ref{fig:hearing_control} illustrates the differences between the considered hearing mechanisms, and shows how they interact with the transmission and the scheduling policies. In the context of WIFI networks, both $\pro$ and $\re$ correspond to the standard 802.11b, which is implemented by using the Carrier-Sense Multiple Access (CSMA) protocol: when considering $\pro$ the CSMA protocol is enhanced with a Request-to-Send/Clear-to-Send (RTS/CTS) mechanism, and when considering $\re$ the CSMA protocol uses a collision detection mechanism~\cite{Karn90}. 

\end{enumerate}

\begin{figure*}[t]
\begin{center}
\subfigure[Proactive hearing control. 1: the transmission policy indicates that the node can transmit at this round. 2: the hearing control transmits a control message to collect available nodes. 3: nodes that hear the message response to the sender. 4: the hearing control asks the scheduler to select a packet from the set of responding nodes. 5: the scheduler selects a packet (if any), passes it to the hearing control and removes it from the queue of packets. 6: the hearing control transmits the packet.]{\includegraphics[scale=0.23]{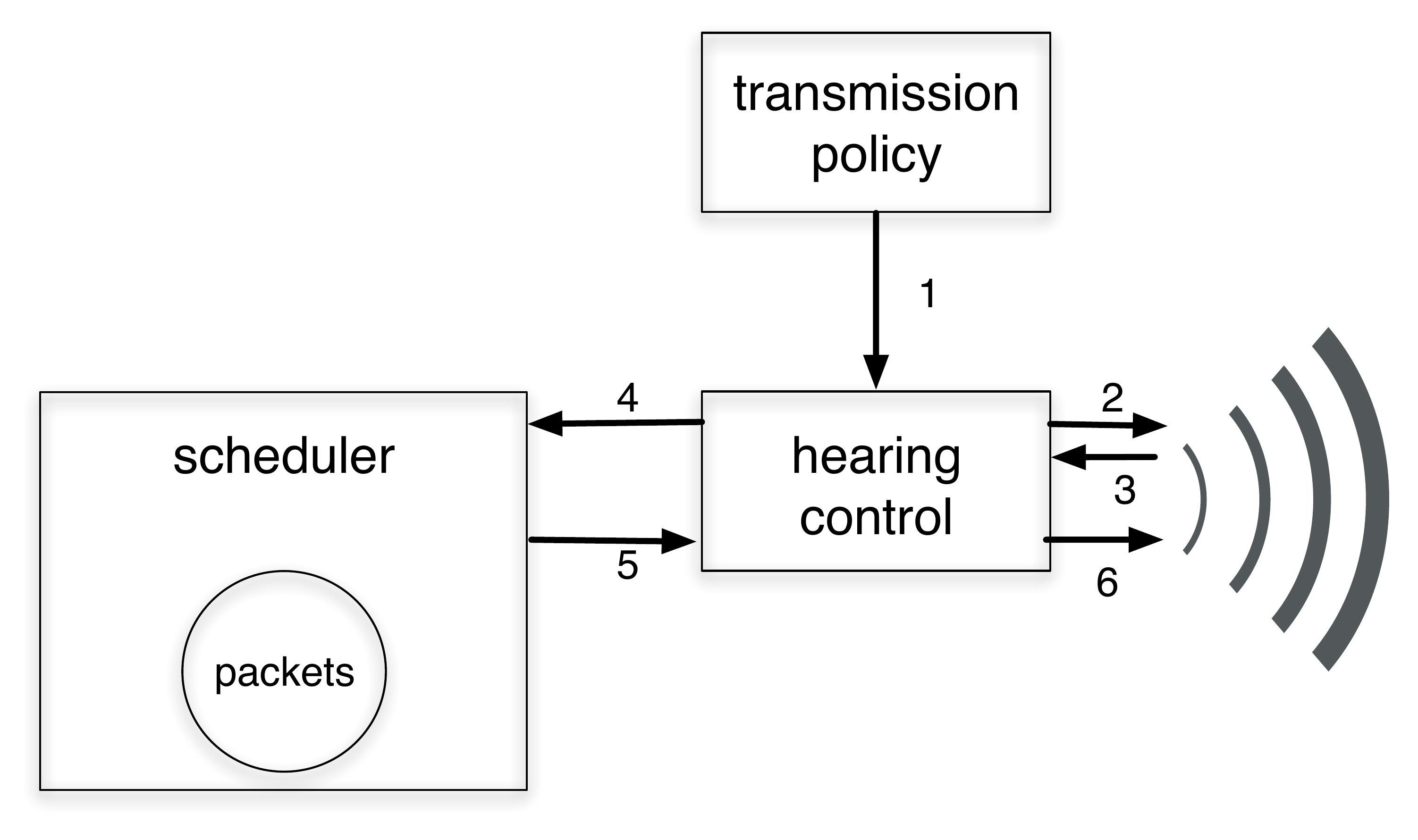}}
\label{fig:pre}
\hspace*{0cm}
\subfigure[Reactive hearing control. 1: the transmission policy indicates that the node can transmit at this round. 2: the hearing control asks the scheduler to select a packet. 3: the scheduler selects a packet (if any) and passes it (a copy) to the hearing control. 4: the hearing control transmits the packet. 5: the hearing control receives an acknowledgement from the physical layer. 6: the hearing control communicates to the scheduler that the packet has been successfully transmitted and the original packet is removed from the queue of packets; otherwise, the original packet is maintained in the queue of packets.]{\includegraphics[scale=0.23]{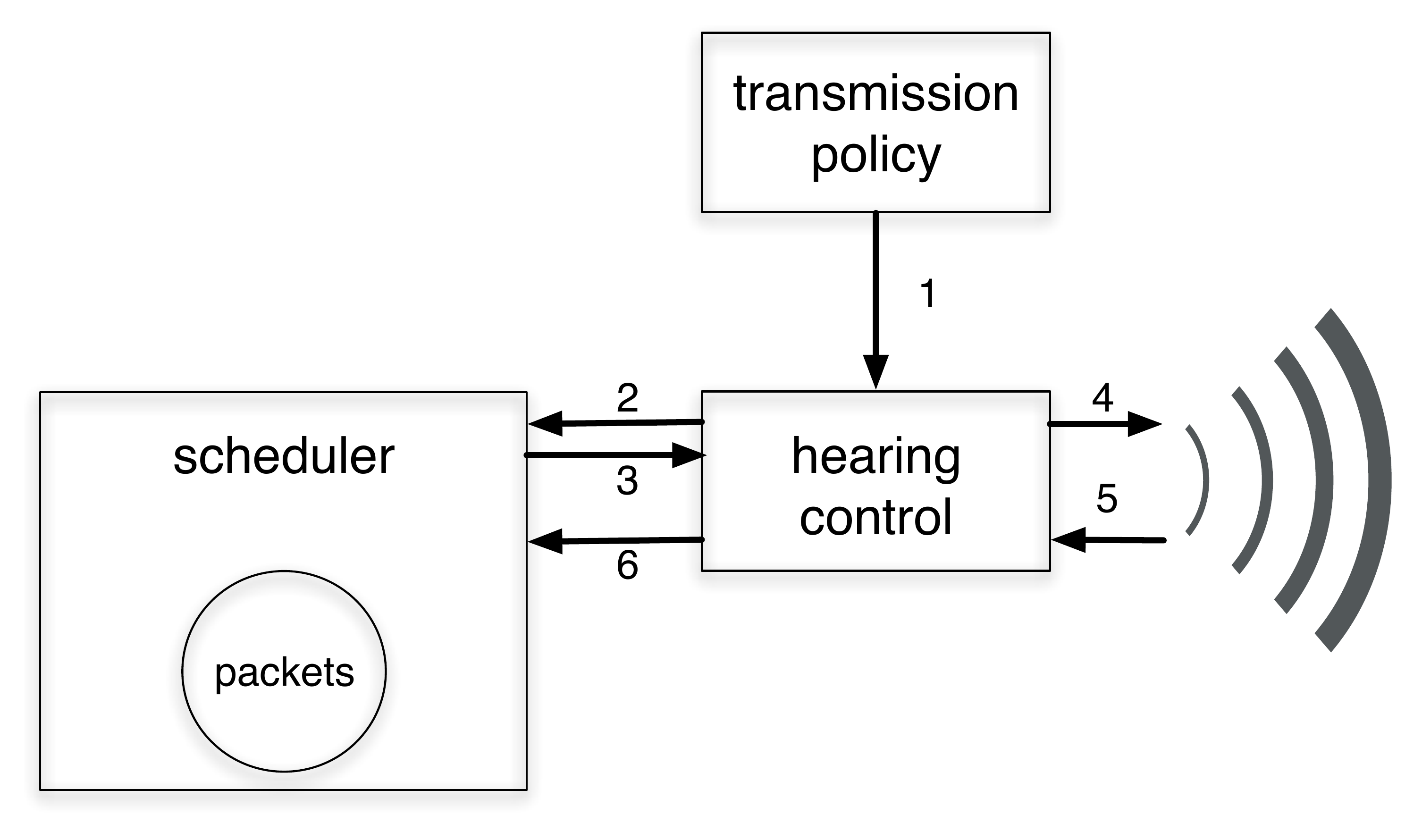}}
\label{fig:post}
\caption{Scheme of the routing protocol at a given node using proactive and reactive hearing control.}
\label{fig:hearing_control}
\end{center} 
\end{figure*}

\subsection{Network Stability}
\label{subsec:stability}

%\paragraph{Adversarial packet injection and stability.}

Packets are injected by adversaries. 
An adversary is determined by a pair of numbers $(b,r)$, called the \emph{type} of the adversary, where \emph{burstiness~$b$} is a positive integer and  \emph{injection rate~$r$} satisfies $0 \leq r \leq 1$.  
We denote by $\cA(b,r)$ an adversary of type~$(b,r)$. 
Such an adversary specifies for each injected packet its complete itinerary.
Let $I(\tau,v)$ represent the number of packets that the adversary injects during time interval $\tau$ and has node~$v$ on its path. 
Adversary $\cA(b,r)$ is constrained such that the inequality 
\[
I(\tau,v) \leq r \; \cdot \mid \tau \mid + \; b
\] 
holds for any $\tau$ and~$v$.
When traffic demands are constrained this way, then we say that they are \emph{admissible for rate~$r$ and burstiness~$b$}.

A routing protocol based on a transmission oracle $\cT$, a scheduling policy~$\cS$ and a hearing control $\cH$ (which can be either $\pro$ or $\re$) is denoted by~$Prot_{\cH}(\cT,\cS)$. 
Let there be given a routing protocol $Prot_{\cH}(\cT,\cS)$ and adversary~$\cA$, and let~$\cD$ be an execution of protocol~$Prot_{\cH}(\cT,\cS)$ against~$\cA$ in a network~$G$. For a positive integer $t$, let $Q_{\cD}(t)$ be the number of packets simultaneously queued in all the nodes in round~$t$ of~$\cD$. 
An execution $\cD$ of a routing protocol is \emph{stable} when  the numbers~$Q_{\cD}(t)$ are all bounded. 
Protocol $Prot_{\cH}(\cT,\cS)$ is \emph{stable  against adversary~$\cA$} if  each  execution of~$Prot_{\cH}(\cT,\cS)$ against~$\cA$ in any network $G$ is stable. Finally, $\cS$ with $\hT_h$ and $\cH$ is \emph{stable against adversary $\cA$} if for any $\cT \in \hT_h$,  protocol $Prot_{\cH}(\cT,\cS)$ is stable against adversary $\cA$.

%\paragraph{Universal stability.}

Given any  transmission providing a link hearing latency $h$, the maximum  injection rate one could expect to guarantee stability is $1/h$. 
Otherwise, instability can be created just by injecting packets passing through a link whose hearing latency is exactly $h$,  at a rate higher than $1/h$.  We say that $\cS$ with $\hT_h$ and $\cH$ is \emph{universally stable} when it is stable for any adversarial injecting rate that is less than or equal to~$1/h$.
A scheduling policy that is not universally stable is called \emph{unstable}.

%----
\section{A general result}
\label{sec:general}

Adversarial queuing was proposed as a methodology to analyze stability in wireline networks~\cite{AndrewsAFLLK01,BorodinKRSW01}. In these and subsequent studies (e.g.,~\cite{AlvarezBS-SICOMP04,AlvarezBDSF05,BorodinOR04,AndrewsFGZ05,CholviE07,EchagueCF03}) it was shown that not all the scheduling policies are universally stable in that setting. We therefore have a natural question: are they universally stable in the wireless model proposed in this paper? In this section we show that the answer to this question is negative. More formally, we show that any scheduling policy that is unstable in the classical wireline adversarial model remains unstable in the multi-hop radio network model, even in scenarios free of interferences. 
 
 Remember that in the multi-hop radio network model, each node maintains a single queue, rather than a dedicated queue for each outgoing link, as in the wireline adversarial queuing model.

 \begin{definition}
Given a network $G$ in a wired scenario, we define its \emph{equivalent network~$G^{\equiv}$} in a wireless scenario as follows: 

(1)~For each link $e$ in~$G$, create a node in $G^{\equiv}$, denoted $v^e$. 

(2)~For each pair of links $e=(-,u)$ and $f=(u,-)$ in $G$, connect $v^e$ to  $v^f$ in~$G^{\equiv}$.
\end{definition}

Observe that, for each queue in $G$, there is a unique queue in $G^{\equiv}$ (and viceversa). The queues in $e$ and in $v^e$ are called \emph{equivalent}, denoted $q$ and $q^{\equiv}$.

In this scenario, let us consider a \emph{work-conserving} transmission oracle (denoted $\WC$) in which every node can transmit in each round using any link, provided there are packets ready to be transmitted. Since there are no interferences, both hearing control mechanisms behave in the same manner (i.e., $\pro$ behaves as $\re$).

\begin{definition}
Given an arbitrary execution $\cD$ in the wired system defined by $(G,\cA,\cS)$, we define the \emph{equivalent} execution $\cD^{\equiv}$ of $Prot_{\cH}(\WC,\cS)$ against adversary $\cA^{\equiv}$ in network $G^{\equiv}$ as follows: for each packet $p$ injected by~$\cA$ at some round, $\cA^{\equiv}$ injects another packet $p^{\equiv}$ at the same round following the path (which we call the $p$-equivalent path) formed by replacing each queue followed by packet $p$ by its equivalent queue. Packet $p^{\equiv}$ may be absorbed at any node pointed by its last traversed queue.
\end{definition}

\begin{lemma}
\label{obs:gT}
The round when an arbitrary packet $p$ is in $q$ in $\cD$ is the same as the round when $p^{\equiv}$ is in $q^{\equiv}$ in $\cD^{\equiv}$.
\end{lemma}
\begin{proof}
For each queue in $G$ there is an equivalent one in~$G^{\equiv}$, so that two queues in~$G^{\equiv}$ are connected provided their equivalent ones are also connected in $G$. 
Any arbitrary packet $p$ follows a given path in $G$, and $p^{\equiv}$ follows the $p$-equivalent path in~$G^{\equiv}$. Since the same scheduling policy and the same work-conserving transmission policy are used in both cases, the result follows.
\end{proof}

\begin{theorem}
\label{the:instability-gT}
If $\cS$ is unstable in the classical wireline adversarial model then $\cS$ is also unstable in the multi-hop radio network model, regardless of the hearing control.
\end{theorem}

\begin{proof}
In order to show that a given scheduling policy $\cS$ is unstable with a transmission oracle $\cT$ and hearing control $\cH$, we need to prove that there is an unstable execution of protocol $Prot_{\cH}(\cT,\cS)$ against some adversary.

Let $\cD$ be an arbitrary execution in $(G,\cA,\cS)$ (in the classical wireline adversarial model). 
By Lemma~\ref{obs:gT}, there is an execution $\cD^{\equiv}$ of $Prot_{\cH}(\WC,\cS)$ against adversary $\cA^{\equiv}$ in network $G^{\equiv}$ such that the round when an arbitrary packet $p^{\equiv}$ is in $q^{\equiv}$ is the same as the round when $p$ is in $q$. Therefore, if $\cD$ is unstable, so is $\cD^{\equiv}$.
\end{proof}

By Theorem~\ref{the:instability-gT} and by the instabilities shown in~\cite{AndrewsAFLLK01,BorodinKRSW01}, we can conclude that FIFO, NTG, FFS and LIFO  are all unstable in the multi-hop radio network model, regardless of the hearing control.

%----
\section{Results by using proactive hearing control}
\label{sec:proactive}

In this section, we address the universal stability property of two well-known scheduling policies (SIS and LIS), when using proactive hearing control. The \emph{Shortest-In-System} (SIS) scheduling policy gives priority to the packet that has been in the system the shortest, with ties broken in an arbitrary manner at each round. The \emph{Longest-In-System} (LIS) scheduling policy gives priority to a packet that has been  longest in the system, with ties also broken in an arbitrary manner at each round.

As it has been pointed in the previous section, we consider that the adversarial injection rate $r$ is always less than or equal to~$1/h$, because $h$ is an upper bound on the hearing latencies of all links. For the proofs, we will use the following terminology: we say that a packet \emph{leaves} a node~$v$ when it is successfully transmitted to the intended neighbor; we also say that  \emph{packet~$p$ has priority over the packet~$q$} if the policy used to assign priorities chooses~$p$ over~$q$. Throughout this section, we consider a system with proactive hearing control and whose transmission oracle is in $\hT_h$. For such a setting, we prove that both scheduling policies are universally stable.

\subsection{Sortest-In-System (SIS)}

\begin{lemma}
\label{lemma:SIS_stability_rng:onepacket}
Consider SIS with a transmission oracle in~$\hT_h$ and hearing control $\pro$. 
For a node $v$ and a packet $p$ in its queue, $v$ will transmit at least one packet with priority higher than that of $p$ during any $h$ rounds during the time interval form $p$'s arrival to $v$ until $p$ is transmitted. 
\end{lemma}
\begin{proof}
Let $e$ be the link through which packet $p$ will be transmitted. 
Let $t$ be the time interval since packet $p$ arrives to $v$ until it is transmitted. 
Recall that the scheduling policy chooses a packet to be transmitted from the set of links that are up at a round.
So each time link $e$ is up in~$t$, a packet with priority over $p$ will be transmitted; otherwise, packet $p$ will be chosen before $t$,  contradicting our assumption.
Since all hearing link latencies are bounded by $h$, then at least one packet with priority over $p$ will be transmitted each $h$ rounds in~$t$.
\end{proof}

\begin{lemma}
\label{lemma:SIS_stability_rng:first}

Let $p$ be a packet waiting in the queue of a node~$v$ at time instant $t_0$, whose scheduling policy is SIS and whose transmission oracle is in $\hT_h$ and the hearing control is $\pro$. Suppose that at this time there are  $k - 1$ other packets in the queue of $v$ that have priority over $p$. 
Then $p$ will leave $v$ within the next $\frac{k+b}{1-r h} \cdot h$ rounds, where $0 \leq r < 1/h$, and $h$ is an upper bound on the hearing latencies of all links.
\end{lemma}

\begin{proof} 
We argue by contradiction. 
Suppose that $p$ does not leave $v$ in the next $\frac{(k+b)}{(1-rh)} \cdot h$  rounds. 
Then other packets different from $p$ must have left the queue meanwhile. 
Because  SIS is the scheduling policy, during that interval the only packets in the system that have priority over $p$ are either those  $k-1$ packets that were present at time $t_0$ or some other that have been injected meanwhile. 

By Lemma~\ref{lemma:SIS_stability_rng:onepacket}, we have that $v$ will transmit at least one packet with priority over $p$ each $h$ rounds until $p$ is transmitted. 
Let us first consider a transmission scenario where only one packet is  transmitted each $h$ rounds. Since only one packet is guaranteed to be transmitted in each interval of $h$ rounds, we have that the following two properties hold: 

(1) the $k-1$ packets currently in the system will take $(k-1) \cdot h$ rounds to leave $v$, and 

(2) the packets injected in the next $\frac{k+b}{1-rh} \cdot h$ rounds, which are $r h \cdot (\frac{k+b}{1-rh} ) + b$, will take at most $(r h\cdot (\frac{k+b}{1-rh}) + b) \cdot h$ rounds to leave $v$.  

Summing up, we obtain that the number of rounds $p$ waits is at most
\begin{gather*}
\hspace*{-65pt}
\bigl(k-1 + r h \cdot \frac{k+b}{1-rh}  + b\bigr) \cdot h \\
 = 
 \Bigl(\frac{k -1 -rhk +rh +rhk + rhb + b -bhr}{1-rh}\Bigr) \cdot h
 \ ,
 \end{gather*}
which is less than $\frac{k+b}{1-rh} \cdot h$.
This results in a contradiction.
We estimated the number of rounds assuming that only one packet is transmitted per interval of $h$ rounds. 
If more than one packet is transmitted per a time interval of $h$ rounds, then this decreases the relative  number of rounds, so that the bound $\frac{k+b}{1-rh} \cdot h$  remains valid.
\end{proof} 

\begin{lemma}
\label{lemma:SIS_stability_rng:second}

Suppose SIS is the scheduling policy with ties broken arbitrarily (i.e., the worst-case solution). 
Define $k_1 = b$ and $k_{i+1}=\frac{k_i + b}{1-rh}$. 
When a packet $p$ arrives at the $i$th queue $v_i$ on its path then there are at most $k_i - 1$ packets requiring any queue in the path of~$p$ with a priority higher than that of~$p$.
\end{lemma}

\begin{proof} 
The proof is by induction on $i$.
Observe that, for any queue $v$, the only packets passing through~$v$ that initially could have priority higher than that of~$p$ are at most $b - 1$ packets injected in the same round as $p$, which provides the base of induction.
To show the inductive step, suppose that the claim holds for some~$i$. 
By Lemma~\ref{lemma:SIS_stability_rng:first}, $p$ will arrive at the tail of $v_{i+1}$ in at most another $\frac{k_i + b}{1-rh} \cdot h$ rounds, during which at most  $r h \cdot (\frac{k_i+b}{1-rh} ) + b$ other packets  requiring any queue $v$ in the path of $p$ with priority over $p$ are injected. 
Thus, when $p$ arrives at the tail of $v_{i+1}$ the number of packets requiring any queue $v$ that have priority higher that that of~$p$ is at most
\begin{gather*} 
k_i-1 + rh \frac{k_i+b}{1-rh}  + b\\
= \frac{k_i -1 -r h k_i + rh + r h k_i + r h b + b -b rh}{1-rh}\\
= \frac{k_i+b}{1-rh} + \frac{rh-1}{1-rh}\\
= k_{i+1}  -1
\ ,
\end{gather*}
so the claim holds.
\end{proof}

\begin{theorem}
\label{the:SIS}
SIS with $\hT_h$ and $\pro$ is universally stable.  No queue contains more than $k_d$ packets, where  $d$ denotes the length of the longest simple directed path in the graph. No packet spends more than $\sum_{i=1}^{d} (\frac{k_i+b}{1-rh}) \cdot h$ rounds in the system.
\end{theorem}

\begin{proof}
We show first that no queue contains more than $k_d$ packets. 
Let us assume that there are $k_d +1$ packets at some point all passing through the same queue. 
By Lemma~\ref{lemma:SIS_stability_rng:second}, the packet with the lowest priority will contradict the property that no queue contains more than $k_i -1$ packets with priority above it. 
Therefore, the overall number of packets is bounded and therefore the system is stable.
Furthermore, combining Lemma~\ref{lemma:SIS_stability_rng:second} with Lemma~\ref{lemma:SIS_stability_rng:first}, we obtain that no packet spends more than $\sum_{i=1}^{d} (\frac{k_i+b}{1-rh}) \cdot h$ rounds in the system.
\end{proof}

\subsection{Longest-In-System (LIS)}
\label{sec:LIS}

For a round~$c$, we denote by \emph{class~$c$} the set of packets injected at round~$c$. 
A class $c$ is said to be \emph{active} at the end of round $t$ if and only if at that round there is some packet in the system of class $c' \leq c$. 
Consider some packet $p$, injected at time $T_0$, and whose path contains queues $v_1, v_2, ..., v_d$, in this order. 
We denote by $T_i$ the round in which $p$ leaves $v_i$, and by $t$  some round in $[T_0,T_d)$. Let $a_t$ denote the number of active classes at the end of round $t$, and define $a = \max_{t \in [T_0,T_d)} a_t$. 
In such a situation, we will  say that $p$ has $a$ active classes while in the system.

\begin{lemma}
\label{lem:Td-minus-T0}
The inequality $T_d - T_0 \leq \frac{(r \cdot a + b) \cdot h \cdot (d-1)}{1 +  r \cdot h \cdot (d-1)}$ holds.
\end{lemma}

\begin{proof} 
Packet $p$ reaches the tail of queue $v_i$ at time $T_i$. 
Since $p$ is still in the system at round $T_i$, all classes formed by packets injected in $[T_0,T_{i-1}]$ are active at the end of that round.
From the definition of $a$, there are at most $a - (T_{i-1} - T_0)$ active classes of packets that can block $p$ in the queue of $v_i$.
As LIS is the scheduling policy, packets injected after $p$ can not block it, because they are in classes after the class of $p$.

Observe that all active classes are consecutive. 
Indeed,  if a class is active then all the subsequent classes are active; so, take the lowest active and all the subsequent classes will be also active.

There are at most $r \cdot (a - T_{i-1} + T_0) + b$ packet in these classes. And since $p$ is one of these packets, at most 

\begin{equation*}
r \cdot (a - T_{i-1} + T_0) + b -1
\end{equation*} 
packets can block $p$. 
Therefore, since $h$ is a bound on the queue latency, we have the following estimates:
\begin{eqnarray}
\label{LIS:ties}
T_i &\leq&  T_{i-1} + (r \cdot (a - T_{i-1} + T_0) + b) \cdot h \\
 &\leq&  T_{i-1} (1 - r \cdot h) + (r \cdot (a + T_0) + b) \cdot h \nonumber \\
 &\leq&  T_{i-1}  + (r \cdot (a + T_0) + b) \cdot h. \nonumber 
\end{eqnarray}
Solving the recurrence results in the following estimate:
\begin{eqnarray*}
T_d &\leq& (r \cdot (a + T_0) + b) \cdot h) \cdot (d-1) + T_0\\
&=& (r \cdot a + b) \cdot h \cdot (d-1) + T_0 \cdot (1 +  r \cdot h \cdot (d-1)).
\end{eqnarray*}
We conclude with this inequality:
$T_d - T_0 \leq \frac{(r \cdot a + b) \cdot h \cdot (d-1)}{1 +  r \cdot h \cdot (d-1)}$.
\end{proof}

\begin{theorem}
\label{the:LIS}
LIS with $\hT_h$ and $\pro$ is universally stable. 
No queue contains more than $r \cdot ((b + r) \cdot h \cdot (d-1) +1) +b$  packets.
No packet spends more than $(b + r) \cdot h \cdot (d-1) +1$ rounds in the system.
\end{theorem}

\begin{proof} 
We show that there are always at most 
\[
(b + r) \cdot h \cdot (d-1) +1
\]
active classes in the system, where $d$  is the length of the longest simple directed path.
Let $a =(b + r) \cdot h \cdot (d-1) +1$ and assume that  the end of round $t$ is the first where there are exactly $a+1$ active classes. 
We show next how to arrive at a contradiction. 
At the end of a round~$a$, there are packets that have been in the system for $a+1$ rounds, and during the first $a$ of these rounds no more than $a$ classes were active.
From Lemma~\ref{lem:Td-minus-T0}, any packet that has at most $a$ active classes while in the system, with a possible exception of the last round, reaches its final destination in a number of rounds that is at most as large as the following estimate:
\begin{gather*}
\frac{(r \cdot a + b) \cdot h \cdot (d-1)}{1 +  r \cdot h \cdot (d-1)} + 1\\
=
\frac{(r \cdot ((b + r) \cdot h \cdot (d-1) +1) + b) \cdot h \cdot (d-1)}{1 +  r \cdot h \cdot (d-1)} + 1\\
=  
(b + r) \cdot h \cdot (d-1) +1
\ .
\end{gather*}
This bound  is less than $ a + 1$, which yields a contradiction.
\end{proof} 

\section{Results by using reactive hearing control}
\label{sec:reactive}

In this section, we show that all the scheduling policies that resolve ties between priorities arbitrarily at each round, are unstable with $\hT_h$ and $\re$, regardless of the injection rate. Examples of such scheduling policies are those that assign priorities based on the packet injection time into the system (e.g., Shorter-In-System and Longest-In-System), or on the traveling path of the packet (e.g., Farthest-To-Go, Nearest-To-Source, Nearest-To-Go and Farthest-From-Source).

\begin{theorem}
\label{thm:break-ties-arbitrarily}
If $\cS$ resolves ties arbitrarily at each round then $\cS$ with~$\hT_h$ and $\re$ is unstable, regardless of the injection rate.
\end{theorem}

\begin{proof} 
Consider a scenario involving three nodes: $u$, $v_1$ and $v_2$, where $u$ is connected to both $v_1$ and~$v_2$. 
Let us inject two packets $p_1$ and $p_2$ at the same time into node $u$ so that $p_1$ is addressed to node $v_1$ and $p_2$ is addressed to node $v_2$. 
Assume that the link $(u,v_1)$ and the link $(u,v_2)$ are up alternately, so that the two links are never both up in the same round. 
Since both packets are injected at the same time into the same node and have one hop to travel, and 
since ties are arbitrarily broken, the scheduling policy can choose any one packet at any round, as far as they are both in node $u$. 
If the scheduling policy chooses $p_2$ when link  $(u,v_1)$ is up, and $p_1$ when link  $(u,v_2)$ is up, then no packet will be successfully transmitted in any round. 
\end{proof}

As a consequence of the previous theorem, a natural question is whether or not instability is due solely to the fact that ties are resolved arbitrarily at each round. In the following theorem,  we also show that instability can be provoked for SIS with $\hT_h$ and $\re$, regardless of how ties are resolved (as a matter of fact, in the proposed unstable scenario no ties need to be resolved).

\begin{figure*}
\centering
\includegraphics[width=\pagewidth]{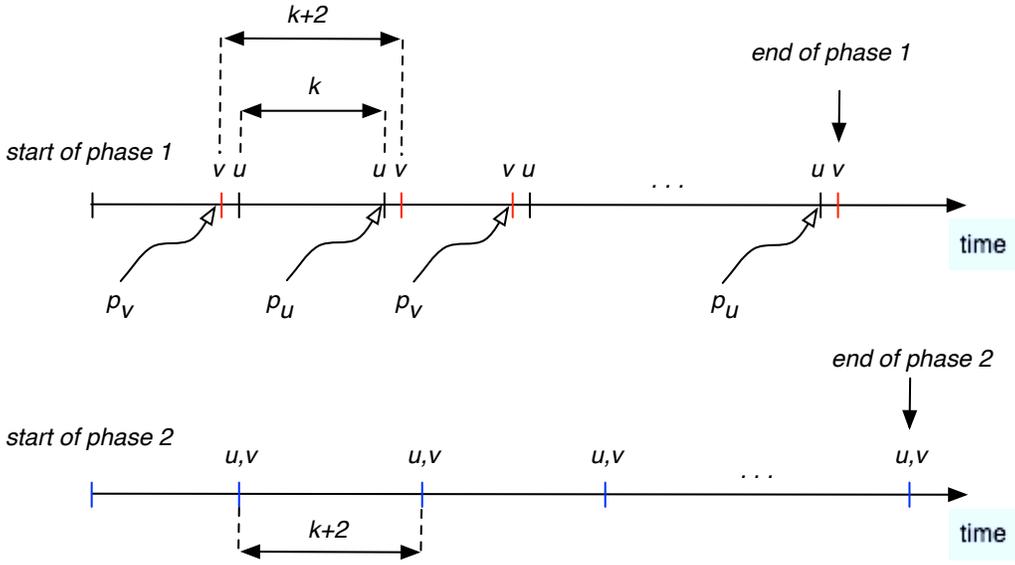}
\caption{A visualization of the execution used in the proof of Theorem~\ref{thm:SIS_inst}.}
\label{fig:SIS}
\end{figure*}

\begin{theorem}
\label{thm:SIS_inst}
SIS with $\hT_h$  and $\re$ is unstable against adversary $\cA(b,1/(2h -4))$, regardless of how ties are resolved, where $h \geq 4$.
\end{theorem}
\begin{proof}

We consider a network topology in which there is a node with two outgoing links to nodes~$u$ and~$v$. 
Packets are injected directly into queues by an adversary $\cA(b,r)$. 
We consider an execution that consists of two phases.
This execution is represented in  Figure~\ref{fig:SIS}.
The phases are specified as follows.

In Phase 1, link $u$ is up each $k$ rounds, and link $v$ is up each $k-1$ and $k+1$ rounds alternately. 
We inject one packet at rounds $k-1, 2k, 3k-1, 4k, 5k-1,\ldots,2b \cdot k$ to traverse the links to $u$ and $v$ alternately, starting with the link to~$v$.
Let us assume that packets that need to traverse the link to~$v$   correspond to the injection-rate component of the adversary's type, and packets that need to traverse the link to~$u$  correspond to the burstiness component of the adversary's type. 

The adversary injects one packet corresponding to the injection-rate component each $2k$ rounds, until the $b$ packets representing the burstiness are injected.
A packet can be transmitted starting from the next round after it has been injected into a queue. 
Because SIS is the scheduling policy, during  Phase~$1$ no packet is transmitted, and at the end of Phase~$1$ there are $2b$  queued packets.

Phase~$2$ phase starts at the same round when Phase~$1$ ends. 
In this phase, both links are up at the same time each $k+2$ rounds for $b$ rounds, and no packet is injected. 
Therefore, at the end of Phase~$2$, some $b$ packets have been transmitted and some $b$ packets remain queued.
Moreover, the adversary can again inject a number of packets corresponding to the burstiness component of its type.

Observe that the latency of the links to $u$ and $v$ is $k+2$, which means that $h = k+2$. Furthermore, the actual injection rate is bounded by $1/2k$, that is, $1/(2h -4)$. 
This, for $h \geq 4$,  is lower or equal than $1/h$, and consequently fulfills the admissibility condition regarding the injection of packets.
At the end of Phase~$2$ we are in the same situation as at the begin of Phase~$1$, except that now $b$ packets remain queued. 
Therefore, we can iterate the same injection pattern to create instability.
\end{proof} 

On the contrary, LIS turns out to be universally stable when ties are resolved, rather than in an arbitrary manner at each round, on a \emph{permanent} one (even if this is arbitrary). We say that ties are resolved in a permanent manner if once a packet $p$ is assigned a higher priority than another one $p'$ at some node, then $p$ will permanently have a higher priority than $p'$ at that node.

\begin{theorem}
\label{thm:LIS_perm}
LIS with $\hT_h$ and $\re$ is universally stable, provided ties are resolved in a permanent manner. No queue contains more than $r \cdot ((b + r) \cdot h \cdot (d-1) +1) +b$  packets.
No packet spends more than $(b + r) \cdot h \cdot (d-1) +1$ rounds in the system.
\end{theorem}
\begin{proof}
Contrary to what happens when ties are arbitrarily resolved at each round, the Equation~(\ref{LIS:ties}) in Lemma~\ref{lem:Td-minus-T0} remains valid when ties are resolved in a permanent manner\footnote{Observe that if ties are resolved arbitrarily at each round then the scenario described in Theorem~\ref{thm:break-ties-arbitrarily} can be used to increase the leaving times unboundedly.}. Therefore, in that case the results in Lemma~\ref{lem:Td-minus-T0} and in Theorem~\ref{the:LIS}  remain valid. 
\end{proof}

\section{Results for regular transmission oracles}
\label{sec:regular}

So far, we have obtained stability results assuming completely general oracles. Indeed, the fact that a scheduling policy $\cS$ with $\hT_h$ and $\cH$ was stable meant that for any $\cT \in \hT_h$,  protocol $Prot_{\cH}(\cT,\cS)$ is stable. However, maybe for some transmission oracles it could be possible to guarantee stability in originally unstable scenarios.

%The results in Theorems~\ref{thm:SIS_inst} and~\ref{thm:break-ties-arbitrarily} don't imply that these scheduling policies are inestable, when using a reactive hearing control, in all cases. For some transmission oracles, universal stability can be guaranteed as when using proactive hearing control.

We denote as $\hT_h^{reg}$ the subclass of oracles in $\hT_h$ (which we call \emph{regular oracles}) in which each node is guaranteed to be able to transmit, using any arbitrary link, at least one packet each $h$ consecutive rounds. If $\cT \in \hT_h^{reg}$ then we say that $\cT$ provides a \emph{node hearing latency} $h$.

\begin{lemma}
\label{fact:in-setting}
If  $\cS$ with $\hT_h$ and $\cH$ is stable against adversary $\cA$ then $\cS$ with $\hT_h^{reg}$ and $\cH$ is stable against $\cA$.
\end{lemma}

\begin{proof} 
Any $\cT \in \hT_h^{reg}$ fulfills that $\cT \in \hT_h$. So, by our assumption, $\cS$ with $\hT_h^{reg}$ and $\cH$ will be  stable against~$\cA$.
\end{proof} 

\begin{lemma}
\label{fact:ack}
If  $\cS$ with $\hT_h^{reg}$ and $\pro$ is stable against an adversary~$\cA$ then $\cS$ with $\hT_{h}^{reg}$  and $\re$ is also stable against $\cA$.
\end{lemma}

\begin{proof}
Let us consider an arbitrary execution $\cD$ of protocol $Prot_{\re}(\cT,\cS)$ against  $\cA$, where $\cT \in \hT_{h}^{reg}$. 
Let us also consider an execution $\cD'$ of protocol $Prot_{\pro}(\cT',\cS)$ against the same adversarial packet injection as in $\cD$, but such that $\cT'$ indicates a node to transmit when such a node successfully transmits in $\cD$; this can also include the case where there are no queued packets. Observe that $\cT'$  provides a node hearing latency of~$h$, since it transmits at the same rounds when $\cT$ transmits. 

We have that nodes transmit the same packets at the same rounds both in $\cD$ and in~$\cD'$. Therefore, the execution $\cD$ is  stable because $\cD'$ is stable. 
\end{proof}

\begin{theorem}
\label{the:node}
SIS and LIS with $\hT_h^{reg}$ and $\re$  are universally stables, regardless of how ties are broken. 
\end{theorem}

\begin{proof}
We use Lemmas~\ref{fact:in-setting} and~\ref{fact:ack}, combined with Theorem~\ref{the:SIS} and Theorem~\ref{the:LIS}.
\end{proof}

\section{Conclusions}
\label{sec:conclusions}

In this article we have introduced a new model to study stability in multi-hop wireless networks in the framework of adversarial queueing. In such a model, a routing protocol consists of three components: a \emph{transmission policy}, a \emph{scheduling policy} to select the packet to transmit form a set of packets parked at a node, and a \emph{hearing control mechanism} to coordinate transmissions with scheduling. Furthermore, the injection of packets is delegated to an adversary that also specifies, for each packet, its complete itinerary.

For such a setting, we have proposed a definition of universal stability that takes into account not only the scheduling policies (as in the standard wireline adversarial model), but also the transmission policies.

\begin{table*}
\centering
\footnotesize
\begin{tabular}{|c||l|l|}\hline
\diaghead{\theadfont Diag ColumnmnHead II}%
{Scheduling\\Policy}{Hearing\\Control}&
\thead{Proactive}&\thead{Reactive}\\    \hline  \hline 
$\cS$ unstable in AQT & unstable & unstable \\   \hline
& & $\bullet$ unstable regardless of how ties are resolved  \\ 
$SIS$ & universally stable  & $\bullet$ universally stable with regular oracles, regardless of  \\
& & how ties are resolved  \\ \hline
&  & $\bullet$ unstable if ties are arbitrarily resolved at each round\\ 
$LIS$ & universally stable & $\bullet$ universally stable if ties are resolved in a permanent manner   \\ 
&  & $\bullet$ universally stable with regular oracles, regardless of \\
& &  how ties are resolved     \\ \hline
\end{tabular}
\caption{Summary of stability results in the multi-hop radio model.}
\label{tb:summary}
\end{table*}

It has been shown that any scheduling policy that is unstable in the classical wireline adversarial model~\cite{AndrewsAFLLK01,BorodinKRSW01} remains unstable in the multi-hop radio network model, even in scenarios free of interferences.

Furthermore, it has been also shown that both SIS and LIS are universally stables in the multi-hop radio network model, provided it is used a \emph{proactive} hearing control. In contrast, such scheduling policies turn out to be unstable when using a \emph{reactive} hearing control. However, the scheduling policy LIS can be enforced to be universally stable provided ties are resolved in a \emph{permanent} manner. Such a situation doesn't hold in the case of SIS, which remains unstable regardless of how ties are resolved. Finally, it has been also shown that all scheduling policies that are universally stable when using  a proactive hearing control (which include SIS and LIS), remain universally stable when using a reactive hearing control, provided \emph{regular} transmission policies are used. Table~\ref{tb:summary} summarizes the obtained results regarding stability in the multi-hop radio model proposed in this paper.

\bibliography{radio-route}

\bibliographystyle{abbrv}

\appendix

%----
\section{Implemented transmission oracles}
\label{sec:transmissions}

In this paper, we have delegated the task of  implementing transmission policies to \emph{transmission oracles} that indicate each node whether or not to transmit in a round. This approach allowed us to abstract from specific transmission policies to consider the qualities of scheduling policies independently of transmission policies. In this section we will show how some of these oracles can be implemented by using transmission policies whose code can be imbedded in the code of the routing protocol (i.e., implemented transmissions). Rather than introducing new policies, we will focus on showing different approaches that can be used to implement them.

\paragraph{Work-conserving.}
Perhaps the simplest transmission policy is the one in which every node transmits at each round using any link (but only one packet per round) as far as  there is a packet ready to be transmitted. In Section~\ref{sec:general}, it has been called \emph{work-conserving} ($\WC$). This transmission policy, provided there are not interferences, has a hearing latency of one round. However, in scenarios with interferences, $\WC$ doesn't guarantee bounded hearing latencies, therefore making the system unstable. 

In the rest of the section, we introduce some implemented transmission policies that provide bounded hearing latencies in scenarios with interferences.

\paragraph{Round-robin.}

A simple transmission mechanism that provides bounded hearing latencies consists of using a token traveling along a logical ring that includes the whole set of nodes. When a node obtains the token, it is eligible to transmit, using any link, for a round. Such a transmission policy, which we call \emph{round-robin} $RR$, provides a hearing latency of $n$, where $n$ is the number of nodes in the system. However, for large values of $n$, this rate is very small and it is not a viable strategy.

%A similar transmission mechanism could consists of using a vertex coloring algorithm.

\paragraph{Transmitters.}

\emph{Transmitters} are mathematical structures that have been used to guarantee successful data transmission within some bounded number of rounds~\cite{DBLP:conf/podc/ChlebusKR06}. Roughly speaking, a transmitter consists of a binary array so that each row indicates a given process when it can transmit (when 1 occurs) or not (when 0 occurs). Furthermore, transmitters guarantee that, for each row, it will happen that at some column a 1 will occur while only 0's occur in this column in the other rows. Therefore, a transmitter of length $m$ guarantees that each node $v$ will be able to successfully transmit using every link outgoing from $v$ within at most $m$ rounds (provided there is a packet pending to be transmitted).

\begin{figure}
\hspace*{4.5cm}\resizebox{0.4\textwidth}{!}{\includegraphics{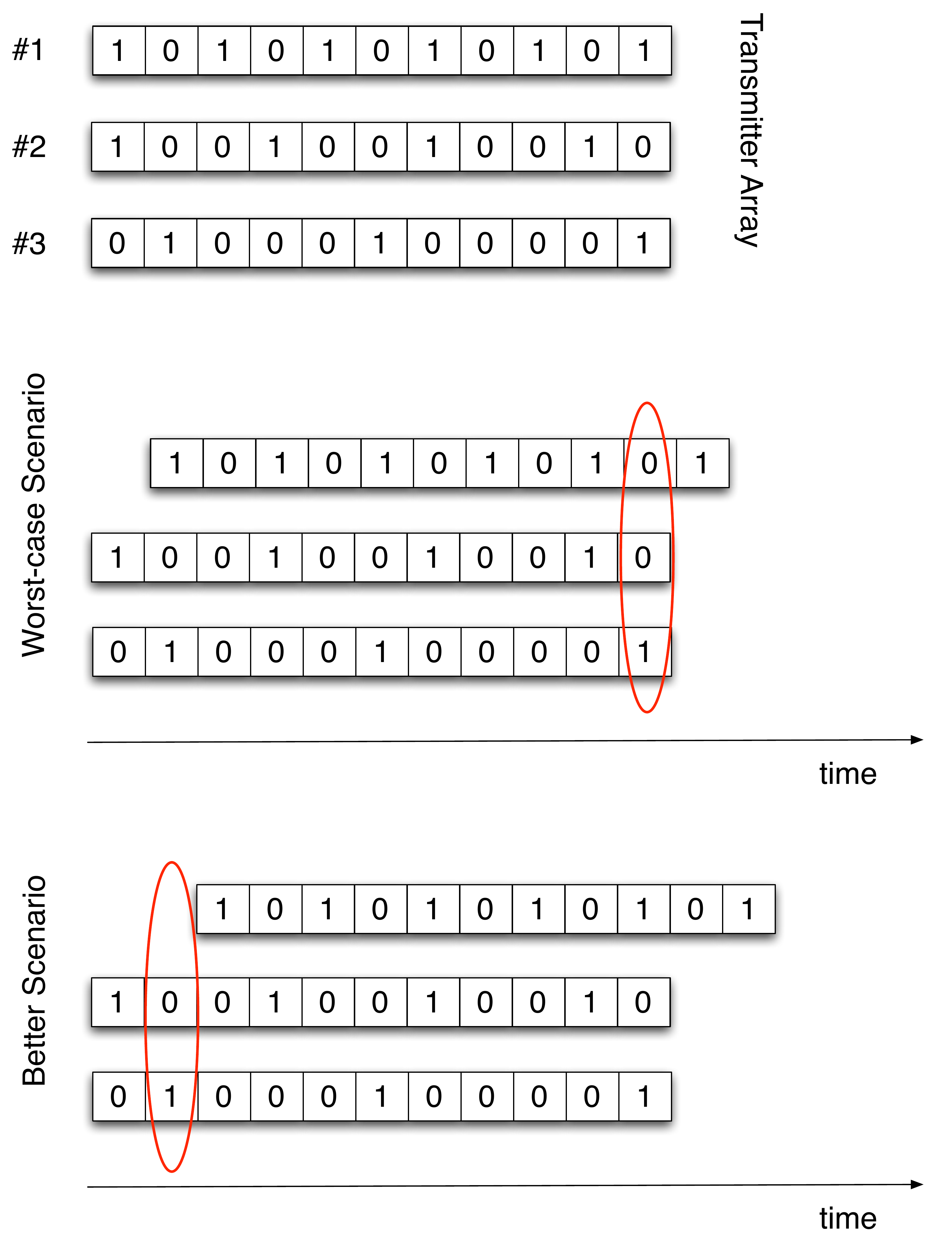}}
\caption{Example of a transmitter array of length $11$ for a system of 3 nodes. Node 3 wants to transmit.}
\label{fig:transmitter}
\end{figure} 

Figure~\ref{fig:transmitter} shows a transmitter array of length $11$ for a system of 3 nodes. The transmitter guarantees that, in the worst case, any node will be able to successfully transmit within the next $11$ rounds. In the worst-case scenario, the node~$3$ transmits after $11$ rounds; but in a better scenario, it transmits after $2$ rounds, and this is the main advantage of transmitters. 

In~\cite{DBLP:conf/podc/ChlebusKR06}, the authors propose a transmitter that provides a node hearing latency of $3 \; n^2 \; \lg n$, where $n$ is the number of nodes (Theorem~8). Also, in~\cite{DBLP:conf/podc/ChlebusK04} they propose a transmitter that provides a node hearing latency of $c  \; k^2 \; \mathrm{polylog} \;n$, where $n$ is the number of nodes, $k-1$ is the maximum node's degree and $c$ is a constant (Corollary~2).

\paragraph{Other implemented transmission oracles.}

In~\cite{FernandezAnta:2012:DRC:2109681.2109778} the authors provide an adaptive transmission policy that guarantees a hearing latency  of $O(k^2 \log k)$, where $k-1$ is the maximum degree of the network (i.e., the maximum number of nodes that can be directly accessed by some node in the system). They also survey several other transmission policies.

\end{document}